\documentclass[pra,showpacs,aps,epsfigure,twocolumn]{revtex4}

\usepackage{array}
\usepackage{amsmath}
\usepackage{graphicx}
\usepackage{dcolumn}
\usepackage{bm}
\usepackage{graphicx}
\usepackage{amsmath}
\usepackage{latexsym}
\usepackage{amsfonts}
\usepackage{amssymb}
\usepackage{array}
\usepackage{epsfig}
\usepackage{epstopdf}
\usepackage{times}
\usepackage{amsthm}

\newtheoremstyle{named}{}{}{\itshape}{}{\bfseries}{.}{.5em}{\thmnote{#3's }#1}
\theoremstyle{named}
\newtheorem*{namedtheorem}{Proposition}

\begin{document}

\title{The work cost of keeping states with coherences out of thermal equilibrium}
\author{Giovanni Vacanti, Cyril Elouard, Alexia Auff\`eves}
\affiliation{Institut Neel, 25 avenue des Martyrs, 38042 Grenoble, France}

\date{\today}

\begin{abstract}
We consider the problem of keeping an arbitrary  state $\rho_s$ out of thermal equilibrium. We find  that counter-acting thermalisation using only a resource system which is in a stationary state at the initial time and a system-resource interaction that preserves the global energy is possible if and only if the target state $\rho_s$ is block-diagonal in the eigenbasis of the system's Hamiltonian $H_s.$ As a consequence, we compute the extra work the operator must provide by tuning the resource-system interaction to overcome this constraint. This quantity, which is interpreted as the work needed to preserve the coherences in the state, can be expressed in terms of the target state $\rho_s$ and the thermal equilibrium state $\rho_{\beta},$ and it is proportional to the symmetrized relative entropy between $\rho_s$ and $\rho_{\beta}.$ \end{abstract}


\maketitle

-\emph{Introduction}:
The idea of interpreting thermodynamic concepts in the context of quantum mechanics has intrigued the physics community for decades. Some of the topics that have been proposed in this area of research in the last years include  quantum applications of Landauer principle\cite{Landauer,ReebWolf,Alexia}, the study of out-of-equilibrium quantum  systems \cite{John,Ross} and the possibility of creating  quantum thermodynamic engines \cite{John2}. Thermodynamics and quantum mechanics are basically involved in all branches of physics, and understanding the interconnections between these two theories is of crucial importance for the resolution of important problems, ranging from the miniaturisation of computer processors \cite{Bennet} to the study of quantum behaviours in out-of-equilibrium biological systems (see \cite{quantumbio} for a recent review on various topics in quantum biology).

On the other hand, in the context of quantum thermodynamics (QT), it is important to determine which states of a given system have genuine quantum features. In this regard, the concept of asymmetry with respect to time translation can be used to distinguish classical states from quantum states\cite{PlenioCoherence}. Loosely speaking, while it might be possible to interprete the thermodynamic properties of states that are symmetric with respect to time translation (i.e. states that  are block-diagonal in the eigenbasis of the system's Hamiltonian) using  the framework of classical statistical mechanics, the thermodynamics of states with non-zero off-diagonal elements (i.e. states with coherences) can only be understood in terms of quantum theory. Notwithstanding the importance of recent results in the field \cite{OscarWorkValue,OscarNegativeEntropy,Kavan,Major,brandao,HoroOppe,Tony}, most of these results focus on the study of block-diagonal states, and  only very recently the role played by  off-diagonal elements has become subject of deeper investigations \cite{TRudolph,TRudolph2,Marvian}.

Following this line of research, here we consider the problem of keeping a system $S$ in a arbitrary state $\rho_s$ while it is in contact with a thermal bath. We assume that the thermalisation process is modeled as a thermalising machine, i.e. a step process in which each step is characterised by a partial swapping between the state of the  system and the state of the bath\cite{ValerioTerm,ValerioHomo}. In the limit of infinitesimal steps, this framework reproduces the continuous dynamics of the system\cite{massimo,ciccarello,ciccarello2}.
We show that the amount of coherence  in the state $\rho_s$ with respect to the eigenbasis of the system's Hamiltonian $H_s$ crucially determines the properties of the operations needed to keep the state out of equilibrium. Specifically, we show that, if we restrict ourself to a certain class of external operations, it is possible to keep $\rho_s$ out of equilibrium if and only if $\rho_s$ is block-diagonal. We also show that operations keeping states with coherences out of equilibrium (thus not belonging to the class of operations mentioned above) require an extra amount of work in order to be implemented, which can be then interpreted as the work needed to keep coherences.


-\emph{Thermalisation processes}: 
We are now going to introduce the concept of \emph{thermalising machines}, which was first proposed in \cite{ValerioTerm,ValerioHomo}. Here, we will make use of thermalising machines in order to give a precise definition of a thermalisation process. This approach is particularly powerful due to the generality of the model, which is applicable to any quantum system.

Let us consider a quantum system $S$ described by a Hilbert space $\mathcal{H}_s$ and with Hamiltonian $H_s.$ The system is  in contact with a thermal bath $B$ at inverse temperature $\beta.$ Following \cite{ValerioTerm,ValerioHomo}, we describe the Hilbert space of the  thermal bath $B$ as a collection of $N$ independent  copies of the system $S$. We now assume that the initial global  state of the system and the bath is given by a product state $\rho_{sb} = \rho_s \otimes \rho_{B},$ with $\rho_s$ an arbitrary density matrix for the system and $\rho_{B} = \bigotimes_{k=1}^{N} \rho_{\beta}^{(k)}$ with $\rho_{\beta}^{(k)} = (e^{-\beta H_s})/(\text{Tr}\{e^{-\beta H_s}\} ).$  Within this model, the thermal bath consists of $N$ identical subsystems, each of which is initialised in the thermal equilibrium state at inverse temperature $\beta$ for  the systems hamiltonian $H_s$.

To characterize the bath-system interaction, we assume a collision model, where the  exchange of information between the system and the bath is described by discrete steps. In each step, the system $S$ interacts only with a single subsystem of the bath. We also assume that the system does not interact more than once with the same copy  throughout the whole  process. The one step interaction is described by the map $\Phi_{\beta}$ acting on the system. The iteration of $\Phi_{\beta}$ eventually takes the system in a state that is arbitrarily close to the thermal equilibrium state after $n$ steps\cite{ValerioTerm}. 

Specifically, we assume now that $\Phi_{\beta}$ is defined in terms of a global partial swapping  unitary  given by   
$P = \cos[\theta] \mathcal{I} + i \sin[\theta] T,$ with $\theta \in ] 0, \pi/2]$ and where $\mathcal{I} = \mathcal{I}_s \otimes \mathcal{I}_b^{(k)}$ is the identity matrix in $\mathcal{H}_s \otimes \mathcal{H}_b^{(k)}$ and $T$ is the total swapping unitary. Without risk of ambiguity, we can simplify the notation by suppressing the index $(k)$. We can now define a thermalisation process as a completely positive trace preserving (CPTP) map $\Phi_{\beta}$ given by $\Phi_{\beta}(\rho_s) = \text{Tr}_{\beta}\{P (\rho_{s} \otimes \rho_{\beta}) P^{\dag}\}.$ It can be easily seen that the local state of the system after the partial trace is given by\cite{ValerioTerm,ValerioHomo}
 \begin{equation}\label{partialswapstates}
\Phi_{\beta}(\rho_{s}) = c^2 \rho_s + s^2 \rho_{\beta} + i cs [\rho_{b},\rho_s]
\end{equation}
where $c=\cos[\theta]$ and $s=\sin[\theta].$ Note that, in our definition, the map $\Phi_{\beta}$ represents a legitimate thermalisation model for any value of $\theta \in ]0, \pi/2].$ For example, a map $\Phi_{\beta}$ corresponding to $\theta = \pi/2$ (i.e. $c=0$ and $s=1$) would thermalise the system in a single step. Clearly, a more realistic model of thermalisation requires a vey small value for $s.$ However, our results are valid for any value of $c$ and $s$ (we only exclude the cases of $\theta =0,$ i.e. $c=1$ and $s=0,$ since $\Phi_\beta$ would trivially leave any state of the system unchanged).

It is worth to note that a thermalising machine described by the map $\Phi_{\beta}$
satisfies two basic requirements of a thermalisation process: 
\begin{enumerate}
\item It is a \emph{thermal operation} \cite{HoroOppe}, i.e. a CPTP map $\Phi(\rho_s) = \text{Tr}_{b}\{U\rho_s\otimes\rho_{\beta} U^\dagger\}$ where (i) the state $\rho_{\beta}$ is a Gibbs state, i.e. $\rho_{\beta} = \frac{e^{-\beta H_{\beta}}}{\text{Tr}\{e^{-\beta H_{\beta}}\}}$, and (ii) $U$ is a global unitary such that $[U, H_s + H_{\beta}] = 0$.
\item  It obeys to zero-law of thermodynamics, in the sense that the iteration of the map will eventually bring the system to its thermal equilibrium state.
\end{enumerate}
Zero-law can be formally written as
\begin{equation}\label{zerolaw}
\forall \rho_s, \forall \epsilon, \exists n_0 :  \{ n \geq n_0\} \Rightarrow \{D_{l_1}(\Phi_{\beta}^n (\rho_s)| \rho_\beta) \leq \epsilon\}.
\end{equation}
where $\Phi_{\beta}^n$ represents $n$ subsequent iterations of the map $\Phi_{\beta}.$ It is important to point out that equation (\ref{zerolaw}) is an essential requirement in our definition of a thermalisation process. 


-\emph{Counter-acting thermalisation}:
The vague concept of acting against thermalisation can be interpreted in different ways. Here we consider a precise type of counter-thermalisation processes, the \emph{restoring processes} and the \emph{stabilising processes}. One simple way to intend the concept of contrasting thermalisation consists of invoking a restoring process. Such operation may be described by a restoring map, i.e. a CPTP-map $\Phi_r$ that allows us to restore the system in a given target state $\rho_s$ at any point during  the thermalisation process. More precisely, given a thermalisation process $\Phi_{\beta}$ whose equilibrium state is $\rho_{\beta},$ we can define the \emph{thermalisation path from $\rho_s$ to $\rho_\beta$ induced by $\Phi_{\beta}$ } as the set of all the states $\rho_s^{\prime}$ such that $\rho_s^{\prime} = \Phi_{\beta}^n(\rho_s)$ for some $n.$ Thus, a restoring map $\Phi_r$ is defined as a map that, after $m$ iteration, transform any state $\rho_s^{\prime}$ belonging to the thermalisation path into the target state $\rho_s.$

The conditions defining  restoring operations may be considered quite demanding, in the sense that the map is required to restore the state $\rho_s$ when acting on  \emph{any} state in the thermalisation path (which include states arbitrarily close to $\rho_{\beta}$). Indeed, one may consider a wider class of maps, which are only able to restore the system after it has gone through a fixed number of thermalisation steps. In particular, we consider the maps that restore the system only if it has gone through a single thermalisation step. Thus we introduce a  second class of counter-thermalising maps, the class of  stabilising maps. We call a stabilising map a CPTP-map which "continuously" keeps a system undergoing a thermalisation process in a given target state $\rho_s.$  More precisely, a \emph{stabilising map with respect to $\rho_s$ and $\Phi_{\beta}$} is a map $\Phi_s : \mathcal{B} (\mathcal{H}_s) \rightarrow  \mathcal{B} (\mathcal{H}_s)$ such that
\begin{equation}\label{stabilisingdef}
\Phi_s(\Phi_{\beta} (\rho_s)) = \rho_s.
\end{equation}

Clearly, the stabilising map will depend on the thermalisation process $\Phi_{\beta}$ and on the target state $\rho_s$. We can very well say that, for any thermalisation model $\Phi_{\beta}$ and any initial state $\rho_s,$ it exists a class of stabilising maps  $\Phi_s,$ and we will denote this class of maps with $\mathcal{S}(\Phi_{\beta},\rho_s).$ Finally, we define the class of all stabilising maps with respect to $\rho_s,$ denoted $\mathcal{S}(\rho_s),$ as the union of all sets $\mathcal{S}(\Phi_{\beta},\rho_s)$ for all possible thermalising machines, i.e. for all possible maps based on the partial swapping operation.


-\emph{Asymmetry and coherence measures}:
At this stage, a precise definition of coherence is needed. 
We define the amount of coherence of a density matrix $\rho,$ with respect to an Hamiltonian $H,$ by using  the concept of symmetric operations under time-translation. Following \cite{Marvian,TRudolph,TRudolph2}, we introduce the class of symmetric operations with respect to an Hamiltonian $H.$ This class includes  any CPTP-map that commute with the time evolution operator generated by $H.$ That is, a symmetric map $\Phi$ is such that $e^{-i H t} \Phi(\rho) e^{i H t}=  \Phi(e^{-i H t} \rho e^{i H t})$.  The set of symmetric states with respect to $H$ are then defined as the set including any state $\rho$ that is invariant under time translations, i.e. $e^{-i H t} \rho e^{i H t} = \rho.$ The symmetric states are thus the states which are diagonal in the eigenbasis of $H$ (for simplicity, here we assume that the Hamiltonian is non-degenerate. All the results can be generalised to the case of a degenerate spectrum with no conceptual difficulties). Naturally, we will call the states which have non-zero off-diagonal elements in the eigenbasis of $H$ asymmetric states. Asymmetric states are then states with coherences, and the words asymmetry and coherence will be used with the same meaning hereafter. 

The amount of coherence in a state can be quantified by any real function $C(\rho)$ which do not increase under symmetric operations, i.e. any function such that  $\forall \Phi$ symmetric, $C(\Phi(\rho))\leq C(\rho).$  In particular, here we will use a specific coherence measure which is based on the $l_1$-distance between Hermitian matrices \cite{PlenioCoherence}. The $l_1$-distance is defined as
\begin{equation}\label{l1def}
D_{l_1}(\rho|\rho^{\prime}) = \sum_{i,j} |\rho_{i,j}-\rho_{i,j}^\prime|.
\end{equation}
The amount of coherence $C_{H}(\rho)$ in a state $\rho$ with respect to an Hamiltonian $H$ is then given by the $l_1$-distance between $\rho$ and its projection on the subspace of symmetric states with respect to $H,$ i.e.
\begin{equation}\label{coherenceDef}
C_{H}(\rho) = D_{l_1}(\rho|P_{\delta}(\rho)) 
\end{equation}
where $P_{\delta}(\rho) = \sum_i \rho_{i,i} |\varphi_i \rangle \langle \varphi_i |,$ with $\{|\varphi_i \rangle\}$ the eigenstates of $H.$ $P_{\delta}(\rho)$ represents the diagonal part of $\rho$ in the basis $\{|\varphi_i \rangle\}$ (the block-diagonal part in case of degeneracy). Thus, the amount of coherence in $\rho$ with respect to $H$ is given by $C_{H}(\rho) = \sum_{i\neq j} |\rho_{i,j}|.$  Notice that, according to this definition, a state is diagonal (block-diagonal for degenerate Hamiltonians) in the eigenbasis of $H$ iff it has zero coherence. 


-\emph{Generalised thermal operations}:
We now introduce the class of  operations that we are allowed to perform. We restrict the class of allowed operations to the class $\mathcal{G}$ of \emph{Generalised Thermal Operations} (GTO) defined as $\Phi(\rho_s) = \text{Tr}_{r}\{U \rho_s \otimes \rho_r U^{\dag}\}$ with
$[U,H_s + H_r] = 0$ and
$[\rho_r,H_r] = 0,$
where $U$ is a unitary operation and $\rho_r$ is the state of an ancillary system $R$ whose Hamiltonian is $H_r.$ This definition represents a generalisation of the definition for thermal operations. Indeed, while we retain the energy conservation condition on the unitary $U,$ we  now allow the use of ancilla states which are not Gibbs states, as long as they are stationary states with respect to the local Hamiltonian $H_r.$ Notice that we do not put any other restriction on the system $R,$ which represents the energy-resource system. This means that we are also allowed to use resource systems with arbitrarily big Hamiltonians.

The motivations behind the choice of considering only GTO will become clear  shortly. 
At this stage we only notice  that the two condition $[U,H_s + H_r] = 0$ and $[\rho_r,H_r] = 0$ are, if taken individually, neutral, meaning that none of the two conditions generates any restriction  on the class of allowed maps if considered alone. This can be easily proven by taking $U$ to be the swapping unitary $T$ and then imposing only one condition at the time. It is easy to see that in both case a map that transforms an arbitrary state $\rho$ into another arbitrary state $\rho^\prime$ can be easily constructed. This is not true for thermal operations. It is also worth to note that GTO are, as well as thermal operations, symmetric operations with respect to $H_s$ \cite{TRudolph}. This means that, any proper coherence measure such as $C_{H_s}(\rho)$ must not increase under GTO. This property plays a crucial role in the following derivation.


-\emph{State's stabilisation and restoring}:
We are now ready to state the main results of this paper. Our first goal is to establish whether it is possible or not to find an operation that contrasts thermalisation, meaning either a restoring or a stabilising operation, which is also a GTO. We find that this is possible if and  only if the target state $\rho_s$ (i.e. the state we want to stabilise or restore) has no coherences. Having introduced the class $\mathcal{S}(\rho_s)$ of stabilising maps with respect to $\rho_s$ and the class  $\mathcal{G}$ of generalised thermal operation, we can state the following proposition: 
\begin{namedtheorem}
A stabilising operation with respect to $\rho_s$ which is also a generalised thermal operation exists if and only if $C_{H_s}(\rho_s)=0$. Formally
\begin{equation}\label{thm1}
C_{H_s}(\rho_s)= 0 \iff \mathcal{S}(\rho_s) \cap \mathcal{G} \neq \emptyset
\end{equation}
\end{namedtheorem}
\begin{proof}
Here, we give a sketch of the proof. The full details of the proof are founded in appendix \ref{appA}. 

In order to prove the left-to-right implication (i.e. zero coherence is a sufficient condition for the existence of the map) we explicitly construct a stabilising map which is also a GTO. Such map is constructed by taking $H_r = H_s,$ $\rho_r=\rho_s$ and the swapping operation as the global unitary $U.$

In order to prove the right-to-left implication (i.e. that zero coherence is a  necessary condition for the existence of the map), first we prove that, for any $\rho_s$ with non-zero coherence, $\Phi_{\beta}$ strictly decreases the coherence of $\rho_s.$ Thus, a stabilising map $\Phi_s$ must necessarily increase the coherence of $\Phi_{\beta}(\rho_s).$ On the other hand, it is proven that GTO cannot increase the amount of coherence in any state\cite{TRudolph}. Hence, it follows that a stabilising map that is also a GTO cannot exist when $C_{H_s}(\rho_s)\neq 0.$
\end{proof}

Clearly, this result is extended to restoring maps too, since these maps are also stabilising maps. Indeed, in the case of restoring maps the result is even more general: it can be easily proven that a restoring operation $\Phi_r$ which is also a GTO cannot exist not only for thermalising machine, but for any model that obey  zero-law of thermodynamics as stated in equation (\ref{zerolaw}) if the target state $\rho_s$ has non-zero coherences. This can be easily understood by noticing that such operations must be able to increase coherences in order to transform states which are arbitrarily close to $\rho_{\beta}$ into $\rho_s.$ Since GTO cannot increase the amount of coherences, it follows that a restoring map $\Phi_r$ cannot be a GTO.  

We now pass to analyse the consequences of proposition (\ref{thm1}). It is clear that, if we are willing to maintain a state with coherences out of thermal equilibrium, we need to abandon one of the two conditions defining GTO. Since the resource system $R$ has the role of an energy reservoir, it is natural to keep the condition $[\rho_r,H_r] = 0,$ which ensures a precise knowledge of the amount of energy available in the reservoir at the initial time. This means that we need to consider unitary operations that do not commute with $H_s + H_r.$ We thus assume that in addition to the energy available in the resource system, the operator is able to change the total energy by tuning the coupling between the resource system $R$ and the system $S$. This energy can be interpreted as work performed by the operator on the system $\{S + R\}$ \cite{Pusz,Alicki,Esposito,Suomela}.

Let us then consider maps $\Phi_s(\rho_s^\prime) = \text{Tr}_{r}\{U \rho_s \otimes \rho_r U^{\dag}\}$ for which $[\rho_r,H_r] = 0$ but any unitary $U$ is allowed. To simplify the notation, here we have set $\rho_s^{\prime} = \Phi_{\beta}(\rho_s).$ The work $W$ done by $U$ on the S+R composed system can be  written as 
\begin{equation}\label{Wdef}
W =\text{Tr} \{(H_s + H_r) U \rho_s^\prime \otimes \rho_r U^{\dag}\} - \text{Tr} \{(H_s + H_r) \rho_s^\prime \otimes \rho_r \}
\end{equation}
The connection between $W$ and the commutativity between $H_s + H_r$ and  $U$ is then clear:  proposition (\ref{thm1}) ensures that,  when the state $\rho_s$ has zero coherence, it is always possible to find a $U$ that commute with $H_s +H_r,$ in which case we have $W=0,$ as can be easily seen from equation (\ref{Wdef}). On the other hand, proposition (\ref{thm1}) also ensures that we must choose an unitary such that $[H_s + H_r, U] \neq 0$ in order to stabilise a state with coherences (assuming that we keep the condition $[\rho_r,H_r] = 0$), in which case we have in general that $W\neq 0.$  The quantity $W$ can be thus interpreted as the work needed in order to preserve coherences from thermalisation. 

To further develop this concept, we now present an example of a map $\Phi_s$ which stabilise  states with coherences and we explicitly calculate the value of $W.$ The stabilising map $\Phi_s$ is constructed  by  considering a resource system composed by many copies of the system itself, all prepared in the target state $\rho_s$ and a global unitary given by the swapping operation $T.$ The map is then given by  $\Phi_{s}(\rho_s^\prime) = \text{Tr}_{r}\{T \rho_s^\prime \otimes \rho_s T^{\dag}\}$ with $T$ the total swapping operator, whose action is given by  $T \rho_s^\prime \otimes \rho_s T^{\dag} = \rho_s \otimes  \rho_s^\prime.$ In order to ensure that $[\rho_s,H_r] =0,$ we pick an Hamiltonian $H_r$ for the resource system given by  
\begin{equation}\label{Hr}
H_r = -(1/\beta) (\log[Z_{H_r}] + \log[\rho_s]).
\end{equation}
with $Z_{H_r} = \text{Tr}_r\{e^{-\beta H_r}\}$ the partition function defined by $H_r$ and the temperature $\beta.$ Clearly, $H_r$ is constructed such that $\rho_s = e^{-\beta H_r}/Z_{H_r}$  is the corresponding Gibbs state at temperature $\beta.$   We can now calculate how much extra work $W$ is required to perform the stabilisation.  Substituting equations (\ref{partialswapstates}) and (\ref{Hr}) in equation (\ref{Wdef}) and writing $H_s$ as $H_s = -(1/\beta) (\log[Z_{H_s}] + \log[\rho_{\beta}])$, it can be shown with a few straightforward passages (see Appendix \ref{appB}) that the work $W$ can be written as 
\begin{equation}\label{CoherenceWork}
W = \frac{s^2}{\beta}D_{\bf symm}(\rho_s | \rho_{\beta})
\end{equation}
where $D_{\bf symm}(\cdot | \cdot)$ is the symmetrized relative entropy defined as $D_{\bf symm}(\rho|\rho^{\prime}) = D(\rho|\rho^\prime) + D(\rho^\prime|\rho),$ with $D(\rho|\rho^\prime) = \text{Tr}\{\rho \log \rho - \rho \log \rho^{\prime}\}$ the usual relative entropy.  Notice that the quantity in Equation (\ref{CoherenceWork}) is always positive, which  means that external work is always required by the system in order to maintain coherences, as expected.



-\emph{Conclusions}:
We have shown that, given a thermalisation model based on partial swapping operations and described by the map $\Phi_{\beta}$, a stabilising map $\Phi_s$ which counteract $\Phi_{\beta}$ and which is, at the same time, a generalised thermal operation exists if and only if the target state $\rho_s$ is block diagonal in the system's Hamiltonian eigenbasis. This implies that, in order to keep a state with coherences out of thermal equilibrium, one of the two conditions defining the class of GTO must be violated. In this regard, we have analysed maps such that $[U, H_s + H_r] \neq 0,$ finding that an extra amount of work $W$ proportional to the symmetrized relative entropy between $\rho_s$ and $\rho_{\beta}$ is required to perform such operations. The work $W$ can be then naturally interpreted as the amount of work that is necessary to perform on the system in order to maintain the  coherences  in $\rho_s$. This result, albeit derived in the context of collision models, can be extended to the case of  a continuous process in which a system is allowed to interact with a bath and a resource system at the same time. Indeed, a continuos process can be simulated by a collision model in the limits of infinitesimal steps\cite{massimo,ciccarello,ciccarello2,Ticozzi}. 


-\emph{Acknowledgements}:
The authors would like to thank Cyril Branciard, Francesco Ciccarello, Maxime Clusel, John Goold, Eduardo Mascarenhas, Massimo Palma and Marcelo F. Santos for helpful discussions. We acknowledge financial support from INCAL project.



\appendix
\begin{widetext}
\section{Full proof of Proposition (\ref{thm1})}
\label{appA}

\begin{namedtheorem}\label{thm2}
A stabilising operation with respect to $\rho_s$ which is also a generalised thermal operation exists iff $C_{H_s}(\rho_s)=0$. Formally
\begin{equation}
\{C_{H_s}(\rho_s)= 0\} \iff \{\mathcal{S}(\rho_s) \cap \mathcal{G} \neq \emptyset\}
\end{equation}
\end{namedtheorem}
\begin{proof}
We first prove the left-to-right implication (Sufficient condition) and then the right-to-left implication (necessary condition).

-\emph{Sufficient condition}:
First we prove that $C_{H_s}(\rho_s) = 0$ is a sufficient condition for the existence of the map, i.e. 
$$\{C_{H_s}(\rho_s) = 0\} \Longrightarrow \{\exists \Phi_s \in  \mathcal{G}\cap \mathcal{S}(\rho_s) \}.$$ 
To do that we explicitly construct a map $\Phi_s(\rho_s^\prime)$ which take any state $\rho_s^\prime$ ( including $\Phi_\beta(\rho_s)$) into $\rho_s$ and that is a  generalised thermal operation. We use an ancillary resource system $R$ whose initial state is $\rho_r$ and whose Hamiltonian is $H_r.$ The map $\Phi_s(\rho_s^\prime)$ is then defined as  
$$\Phi_s(\rho_s^{\prime}) = \text{Tr}_r \{T \rho_s^{\prime} \otimes \rho_r T^\dag\},$$ 
where the unitary $T$ is the total swapping operation whose action is given by  
$$T \rho_s^\prime \otimes \rho_r T^{\dag} = \rho_r \otimes \rho_s^\prime.$$
 By choosing $\rho_r = \rho_s$ and $H_r=H_s,$ $\Phi_s$ is a then a  stabilising operation with respect to $\rho_s$ and it can be easily checked that it is also a generalised thermal operation (the swapping operation is energy preserving since $H_r = H_s$ and $[\rho_s,H_s] = 0$ since the state $\rho_s$ is bock diagonal in the eigenbasis of $H_s$).

-\emph{Necessary condition}: we now prove that $C_{H_s}(\rho_s) = 0$ is also a necessary condition for the existence of a stabilising map, i.e.
 $$\{C_{H_s}(\rho_s) = 0\} \Longleftarrow \{\exists \Phi_s \in  \mathcal{G}\cap \mathcal{S}(\rho_s) \}$$

 \emph{Step 1}: As a first step we prove that  for any map $\Phi_{\beta}$ written in the form $\Phi_{\beta}(\rho_s) = \text{Tr}_{\beta}\{P (\rho_{s} \otimes \rho_{\beta}) P^{\dag}\},$ with $P = \cos[\theta] \mathcal{I} + i \sin[\theta] T$ and $\theta \in ] 0, \pi/2],$ and  for any state $\rho_s$ with non-zero coherence, we have that  
 \begin{equation}\label{firststep}
 C_{H_s}(\Phi_{\beta}(\rho_s)) < C_{H_s}(\rho_s).
 \end{equation}
  To do so, we explicitly  write the state $\Phi_{\beta}(\rho_s),$ which according to equation (\ref{partialswapstates}) is given by 
$$\Phi_{\beta}(\rho_{s}) = c^2 \rho_s + s^2 \rho_{\beta} + i cs [\rho_{\beta},\rho_s].$$ We now pick the basis $\{|i\rangle\}$ of eigenvectors of $H_s$  in which $\rho_{\beta}$ is diagonal (although here we consider non-degenerate Hamiltonians for the sake of simplicity, choosing such basis is necessary to extend the prove to the case in which  $H_s$ has degenerate eigenvalues). Following the definition of coherence given in Equation (\ref{coherenceDef}) (i.e. the sum of the modulus of the off-diagonal elements in $\rho$) we  can easily derive that
\begin{equation}
C_{H_s}(\Phi_{\beta}(\rho_s)) = \sum_{i\neq j} \big|c^2 \rho_s^{(i,j)} +i cs (\rho_{\beta}^{(i,i)} - \rho_{\beta}^{(j,j)}) \rho_s^{(i,j)} \big|
\end{equation}
where $\rho^{(i,j)} = \langle i | \rho | j\rangle$ are the matrix elements of the two density matrices in the basis $\{|i\rangle\}.$
Notice that, since $\rho_{\beta}^{(i,i)}$ represents probabilities, we have that $(\rho_{\beta}^{(i,i)} - \rho_{\beta}^{(j,j)})^2 \leq 1.$ Taking into account that we consider maps with $\theta \in ] 0, \pi/2],$ which means $c\in [0,1[$ and $s \in ]0,1],$ we can write\begin{equation}
\begin{split}
C_{H_s}(\Phi_{\beta}(\rho_s)) = &\sum_{i\neq j} \big|c^2 \rho_s^{(i,j)} +i cs (\rho_{\beta}^{(i,i)} - \rho_{\beta}^{(j,j)}) \rho_s^{(i,j)} \big| \\ 
= &\sum_{i\neq j} \big|c \rho_s^{(i,j)}| | c+ i s (\rho_{\beta}^{(i,i)} - \rho_{\beta}^{(j,j)}) \big|\\
=& \sum_{i\neq j} c |\rho_s^{(i,j)}| \sqrt{c^2 + s^2 \big(\rho_{\beta}^{(i,i)} - \rho_{\beta}^{(j,j)}\big)^2} \leq  \sum_{i\neq j} c |\rho_s^{(i,j)}| \sqrt{c^2 + s^2} \\ 
= & c \sum_{i\neq j} |\rho_s^{(i,j)}| = c \cdot C_{H_s}(\rho_s)
\end{split}
\end{equation}
Since $c  < 1$ for any non trivial process, we have that
\begin{equation}
C_{H_s}(\Phi_{\beta}(\rho_s)) \leq  c \cdot C_{H_s}(\rho_s) < C_{H_s}(\rho_s)
\end{equation}
This prove equation (\ref{firststep}).  

\emph{Step 2}: As a second step, we notice that generalised thermal operation cannot increase the amount of Coherence in any state $\rho_s$, i.e. for any generalised thermal operation $\Phi_s$ we have
\begin{equation}\label{GTOCoher}
C_{H_s}(\Phi_s(\rho_s)) \leq C_{H_s}(\rho_s)
\end{equation}
 This property is proven in \cite{TRudolph2}.

\emph{Step 3}: The last step of the proof is straightforward. Indeed, a stabilising map is such that $\Phi_s(\Phi_{\beta}(\rho_s)) = \rho_s.$ Because of (\ref{firststep}), $\Phi_{\beta}$ must strictly decrease the coherence for any state.  Therefore, the stabilising map $\Phi_s$ must necessarily increase the amount of coherence in $\Phi_{\beta}(\rho_s).$ Because of (\ref{GTOCoher}), this means that it cannot be a generalised thermal operation.
 
\end{proof}


 \section{Derivation of  Equation (\ref{CoherenceWork})}
 \label{appB}
We start by considering the expression in equation (\ref{Wdef}). Keeping in mind that we have
 \begin{equation}
 \begin{split}
 &\rho_r = \rho_s, \\ 
 &\rho_s^{\prime} = \Phi_{\beta}(\rho_s) =c^2 \rho_s + s^2 \rho_{\beta} + ics [\rho_{\beta},\rho_s], \\ 
  & U = T \quad \text {with T representing the swapping unitary}\\
 &H_s = -(1/\beta) (\log[Z_{H_s}] + \log[\rho_{\beta}]) \quad \text{with} \quad Z_{H_s} = \text{Tr}_s\{e^{-\beta H_s}\}, \\
&H_r = -(1/\beta) (\log[Z_{H_r}] + \log[\rho_{s}]) \quad \text{with} \quad Z_{H_r} = \text{Tr}_s\{e^{-\beta H_r}\},
 \end{split}
 \end{equation}
  equation (\ref{Wdef})  can be transformed according to
 \begin{equation}
 \begin{split}\label{Equstep1}
 W = &\text{Tr} \{(H_s + H_r) U \rho_s^\prime \otimes \rho_r U^{\dag}\} - \text{Tr} \{(H_s + H_r) \rho_s^\prime \otimes \rho_r \}=\\
 &\text{Tr} \{(H_s + H_r) T \rho_s^\prime \otimes \rho_s T^{\dag}\} - \text{Tr} \{(H_s + H_r) \rho_s^\prime \otimes \rho_s \} \\
 = &\text{Tr} \{(H_s + H_r) \rho_s \otimes  \rho_s^\prime\}   - \text{Tr} \{(H_s + H_r) \rho_s^\prime \otimes \rho_s \}  \\ 
 =& \text{Tr}_s \{ H_s (\rho_s - \rho_s^\prime)\} + \text{Tr}_r \{ H_r ( \rho_s^\prime - \rho_s )\}\\
 = &\text{Tr}_s \{ H_s (\rho_s - c^2 \rho_s - s^2 \rho_{\beta} - ics [\rho_{\beta},\rho_s])\} + \text{Tr}_r \{ H_r (c^2 \rho_s + s^2 \rho_{\beta} + ics [\rho_{\beta},\rho_s] - \rho_s)\}\\
 = &s^2 \Big[\text{Tr}_s \{ H_s (\rho_s -\rho_{\beta}) \} - \text{Tr}_r \{ H_r (\rho_s -\rho_{\beta} )\}\Big]- i cs \Big[\text{Tr}_s \{ [H_s,\rho_{\beta}] \rho_s \} + \text{Tr}_r \{ \rho_{\beta}[\rho_{s},H_r]\} \Big].
 \end{split}
 \end{equation}
Noticing that $[H_s,\rho_{\beta}] = [H_r,\rho_{s}] = 0,$ the second term in the last line of equation (\ref{Equstep1}) vanishes. After substituting the expressions for $H_s$ and $H_r$ in the first term, we have
\begin{equation}
 W = \frac{s^2}{\beta} \Big[- \log Z_{H_s} \text{Tr}_s \{(\rho_s -\rho_{\beta}) \} + \log Z_{H_r} \text{Tr}_r \{(\rho_s -\rho_{\beta}) \} - \text{Tr}_s \{(\rho_s -\rho_{\beta}) \log \rho_{\beta}\}   + \text{Tr}_r \{(\rho_s -\rho_{\beta}) \log \rho_{s}\} \Big].
 \end{equation}
Since $\rho_s$ and $\rho_{\beta}$ have are both trace-1 matrices, the first two terms in the expression above vanish. Re-arranging the two remaining terms and considering that the system $r$ is a copy of the system $s,$  so that  tracing over $s$ and tracing over $r$ are equivalent, we obtain 
 \begin{equation}
 \begin{split}
 W = &\frac{s^2}{\beta} \Big[- \text{Tr}_s \{(\rho_s -\rho_{\beta}) \log \rho_{\beta}\}   + \text{Tr}_r \{(\rho_s -\rho_{\beta}) \log \rho_{s}\} \Big] = \\ 
 =& \frac{s^2}{\beta} \Big[ \text{Tr} \{ \rho_s \log \rho_s - \rho_s \log \rho_{\beta}\} + \text{Tr} \{ \rho_{\beta} \log \rho_{\beta} - \rho_{\beta} \log \rho_{s}\} \Big] = \\
 = &\frac{s^2}{\beta} [D(\rho_s| \rho_{\beta}) + D(\rho_{\beta}| \rho_s)].
  \end{split}
 \end{equation}
This  is the expression given in equation (\ref{CoherenceWork}) for $W.$
 
 \end{widetext}

\end{document}